\documentclass[a4paper,UKenglish,cleveref, autoref]{lipics-v2021}
\usepackage{amsmath}
\usepackage{amsfonts}
\usepackage{amssymb}
\usepackage{xspace}
\usepackage[table,xcdraw]{xcolor}
\usepackage{graphicx}
\usepackage{multicol}
\usepackage{stmaryrd} %llbrackets

\usepackage{tikz}
\usetikzlibrary{arrows,automata,decorations.pathmorphing}
\usetikzlibrary{positioning}
\usetikzlibrary{shapes,shapes.geometric,fit,calc,automata,}

\newcommand{\NotNeeded}[1]{}
\newcommand{\FullVersion}[1]{}
\newcommand{\NotNow}[1]{}

\newcommand{\Subject}[1]{\subparagraph{#1.}}
\newcommand{\St}{~|~}
\newcommand{\con}{\cdot}

\newcommand{\Nat}{\ensuremath{\mathbb{N}}\xspace}
\newcommand{\Rat}{\ensuremath{\mathbb{Q}}\xspace}
\newcommand{\Reals}{\ensuremath{\mathbb{R}}\xspace}

\newcommand{\out}[1]{\bar{#1}}

\newcommand{\best}[2]{\mathtt{Best}_{#1}(#2)}

\newcommand{\True}{\mathtt{true}}
\newcommand{\False}{\mathtt{false}}

\newcommand{\Func}[1]{{\mathsf{#1}}}
\newcommand{\BP}{\Func{B}^+}

\newcommand{\A}{{\cal A}}
\newcommand{\B}{{\cal B}}

\newcommand{\T}{{\cal T}}

\newcommand{\trans}[3]{#1\xrightarrow[]{#2}#3}
\newcommand{\letter}{a}
\newcommand{\weight}{x}

\newcommand{\Val}{\Func{Val}}
\newcommand{\Inf}{\Func{Inf}}
\newcommand{\Sup}{\Func{Sup}}
\newcommand{\Sum}{\Func{Sum}}
\newcommand{\DSum}{\Func{DSum}}
\newcommand{\Avg}{\Func{Avg}}
\newcommand{\LimInf}{\Func{LimInf}}
\newcommand{\LimSup}{\Func{LimSup}}
\newcommand{\LimInfAvg}{\Func{LimInfAvg}}
\newcommand{\LimSupAvg}{\Func{LimSupAvg}}
\newcommand{\MPL}{\underline{\Func{MeanPayoff}}}
\newcommand{\MPH}{\overline{\Func{MeanPayoff}}}

\definecolor{darkgreen}{rgb}{0.13, 0.55, 0.13}

\newcommand{\CheckRelevance}[1]{}

\newcommand{\LW}[2]{#1\!:\!#2}
\newcommand{\strat}{s}

\renewcommand{\phi}{\varphi}
\renewcommand{\epsilon}{\varepsilon}

\title{History Determinism vs. Good for Gameness\newline in Quantitative Automata} 

\titlerunning{History Determinism vs. Good for Gameness}%optional, please use if title is longer than one line
\author{Udi Boker}{Reichman University, Herzliya, Israel}{}{}{}

\author{Karoliina Lehtinen}{
	CNRS, Marseille-Aix Universit\'e, Universit\'e de Toulon, LIS, Marseille, France}{}{}{}%{lehtinen@lis-lab.fr}{}{}

\authorrunning{U. Boker and K. Lehtinen}%TODO mandatory. First: Use abbreviated first/middle names. Second (only in severe cases): Use first author plus 'et al.'

	\Copyright{Udi Boker and Karoliina Lehtinen}%TODO mandatory, please use full first names. LIPIcs license is "CC-BY";  http://creativecommons.org/licenses/by/3.0/


%TODO mandatory; please add comma-separated list of keywords

\category{}%optional, e.g. invited paper

\supplement{}%optional, e.g. related research data, source code, ... hosted on a repository like zenodo, figshare, GitHub, ...

\acknowledgements{We thank Jan Otop for discussing Borel definability in quantitative automata.}%optional

\nolinenumbers %uncomment to disable line numbering

\EventEditors{Miko{\l}aj Boja\'{n}czyk and Chandra Chekuri}
\EventNoEds{2}
\EventLongTitle{41st IARCS Annual Conference on Foundations of Software Technology and Theoretical Computer Science (FSTTCS 2021)}
\EventShortTitle{FSTTCS 2021}
\EventAcronym{FSTTCS}
\EventYear{2021}
\EventDate{December 15--17, 2021}
\EventLocation{Virtual Conference}
\EventLogo{}
\SeriesVolume{213}
\ArticleNo{35}
\begin{document}

\maketitle

\begin{abstract}
Automata models between determinism and nondeterminism/alternations can retain some of the algorithmic properties of deterministic automata while enjoying some of the expressiveness and succinctness of nondeterminism.
We study three closely related such models -- history determinism, good for gameness and determinisability by pruning -- on quantitative automata.

While in the Boolean setting, history determinism and good for gameness coincide, we show that this is no longer the case in the quantitative setting: good for gameness is broader than history determinism, and coincides with a relaxed version of it, defined with respect to thresholds. We further identify criteria in which history determinism, which is generally broader than determinisability by pruning, coincides with it, which we then apply to typical quantitative automata types.

As a key application of good for games and history deterministic automata is synthesis, we clarify the relationship between the two notions and various quantitative synthesis problems. We show that good-for-games automata are central for ``global'' (classical) synthesis, while ``local'' (good-enough) synthesis reduces to deciding whether a nondeterministic automaton is history deterministic.
\end{abstract}

\section{Introduction}\label{sec:intro}
Boolean automata recognise languages of finite or infinite words, often used in verification to describe system behaviours. In contrast, quantitative automata define functions from words to values, and can describe system properties such as energy usage, battery-life or costs. Like Boolean automata, quantitative automata can have nondeterministic choices (disjunctions) and universal choices (conjunctions), which make them more powerful than deterministic models. Alternating automata combine both nondeterministic and universal choices.

However, not all nondeterminism is born equal. Generally, nondeterminism increases the expressiveness and succinctness of an automata model, but at the cost of also increasing the complexity of algorithmic problems on it, sometimes even rendering them undecidable. However, restricted forms of nondeterministic and even alternating automata can enjoy some of the good algorithmic properties of deterministic automata while also gaining in expressiveness and succinctness. 

We focus on three closely related restrictions on nondeterminism and alternations, relevant to the synthesis problem. \textit{History determinism}~\cite{Col09} postulates that the choices in the automaton -- whether they be nondeterministic or universal -- should not depend on the future of the input word. That is, one should be able to construct runs letter by letter while reading the input word, so that the resulting run is as good as one constructed with the knowledge of the full word. 
%We will see later what meanings \textit{good enough} can take in the quantitative setting. 
The notion of \textit{good for games} automata comes from solving two-player games without determinisation~\cite{HP06}. It postulates that the composition of such an automaton $\A$ with games whose payoff function is described by $\A$ should be an equivalent game -- that is, one with the same winner in the Boolean setting, or the same value in the quantitative setting. Finally, an automaton is \textit{determinisable by pruning} if it embeds an equivalent deterministic automaton and, at least in the nondeterministic case, this notion can be seen as a (stronger) ``semi-syntactic'' version of history determinism.

The three notions are well studied in the Boolean setting. There, history determinism and good for gameness coincide, and are broader than determinisability by pruning in general, but coincide with it for some automata types \cite{BL19}.

We generalize these notions to the quantitative setting and study the relations between them. Some versions of these notions already appear in the literature with respect to quantitative automata, as we elaborate on in the related-work paragraph, however not in a systematic and consistent way, and without analysis of the relations between them.

We start with general results concerning arbitrary quantitative automata and then provide a more specific analysis of the following most common types of quantitative automata: $\Sum, \Avg, \Inf, \Sup$, discounted sum ($\DSum)$, $\LimInf, \LimSup, \LimInfAvg$ and $\LimSupAvg$.

Surprisingly, it turns out that good for gameness and history determinism no longer coincide in the quantitative setting.
The surprise comes from the fact that the two names are used interchangeably in the Boolean setting and are already starting to mix in the quantitative setting. 
(In the Boolean setting, even the seminal paper of Henzinger and Piterman \cite{HP06}, which named the ``good for games'' notion, defined history deterministic automata and showed that they are indeed good for games, while the other direction was only shown later~\cite{BL19}. In the quantitative setting, \cite{HPR15,HPR16,HPR17} speak of good for games quantitative automata, although their definition is closer to history determinism.)

We first observe that in the quantitative setting, the three notions need sub-notions, relating to whether one considers automata/games equivalence with respect to values or thresholds. (See \cref{sec:QuantitativeGfgDefinitions} for the exact definitions.)

We then show that while good for gameness coincides with threshold good for gameness, history determinism is stricter than threshold history determinism, and only the latter, under some assumptions, is equivalent to good for gameness. (See \cref{fig:Comaprison} for a detailed scheme of the relations.) The assumption for the equivalence of threshold history determinism and good for gameness is that the ``letter game'' played on the quantitative automaton (which defines whether or not it is history deterministic) is determined. We show that this is guaranteed for quantitative automata whose threshold versions define Borel sets.

Determinizability by pruning, which has an appealing structural definition, is generally stricter than history determinism for nondeterministic automata, already in the Boolean setting, while equivalent to it for some automata types.
We observe that the two notions are incomparable for alternating automata, already in the Boolean setting (see \cref{fig:DbpInCmpHd}).
We then analyse general properties of value functions that guarantee the equivalence of determinizability by pruning and history determinism for all nondeterministic quantitative automata whose value function has these properties.
We  apply these results to specific automata types. Specifically, we show the equivalence for 
$\Sum, \Avg, \Inf$ and $\Sup$ automata on finite words and $\DSum$ automata on finite and infinite words.

Finally, we discuss how the different notions are relevant for different quantitative synthesis problems.
%Recall that the Boolean reactive synthesis problem asks whether there is a system that can  realize the specification language, in the sense that for every input sequence, it produces an output sequence, letter by letter, such that the interleaved sequence, representing the global behaviour, is in the specification.
In quantitative synthesis~\cite{brenguier2016non,SKRV17}, the specification is a function $f$ that maps sequences of input-output pairs onto values. The goal of the system is to respond to input letters by producing output letters while maximising the value of the resulting input-output sequence.
Given a function $f$, one can ask several questions: (i) what is the best value a system can guarantee over all inputs~\cite{RCHJ09}? (ii) can it guarantee at least a threshold value? (iii) can it guarantee for each input sequence $I$ the best value that an input-output sequence including $I$ has~\cite{FLW20}? (iv) can it achieve a threshold value $t$ for all inputs that appear in an input-output sequence with value at least $t$?
%In each case, there are two computational problems: deciding whether a solution exists, called \emph{realisability}, and producing a solution function, typically as a transducer. These problems are often further parameterised by the computational power of the required output transducer, which is tightly linked to the complexity of the specification.
%In~\cref{sec:synthesis} we discuss how these different synthesis problems relate to good for gameness and history determinism. 
In a nutshell, we show that on one hand, (threshold) good for games alternating quantitative automata can be used to solve (i) and (ii) via a product construction similar to the one used for deterministic automata~\cite{RCHJ09}; and on the other hand, (iii) and (iv) for (threshold) history deterministic nondeterministic automata are linearly inter-reducible with deciding the (threshold) history-determinism of an automaton.

\Subject{Related work}
Thomas Colcombet’s original definition of history determinism \cite{Col09} also considered non-Boolean automata, namely cost automata. While the restriction of his definition to $\omega$-regular automata coincides with the original definition in \cite{HP06} of good for games automata \cite{BL19}, in the quantitative setting his definition is different from what we provide here. His notion can be viewed as `approximated history-deterministic with respect to a threshold' as it asks for an approximation ratio that describes the difference between the value achieved by a strategy without the knowledge of the full input word and the actual value of the word. %: an automaton $\A$ satisfies Colcombet’s definition if there is an approximation ratio $\alpha$, such that for every threshold $n$, there is a strategy $\delta_n$ that achieves a value up to $\alpha$ worse than the value that full nondeterminism provides, with respect to words whose value with full nondeterminism is below the threshold $n$.
Another notion of approximative history determinism appears in~\cite{HPR15,HPR16,HPR17} under the name of $r$-GFGness, where $r$ is a bound on the difference of the two values.
Zero-regret determinizability~\cite{AKL10,HPR16} on the other hand lies somewhere between approximative determinizability by pruning and approximative history determinism. It requires an automaton to be approximatively equivalent to a deterministic automaton obtained by taking the product of the input automaton with a finite memory, with both the size of the memory and permitted regret as parameters. When both are set to zero, we have determinizability by pruning.

Observe that we use the term ``quantitative automata'' rather than ``weighted automata''.  
The latter usually relates to the algebraic definition, whereby the value of a nondeterministic automaton on a word is the semiring sum (or valuation-monoid sum) of its accepting runs' values. It is generally not defined for alternating automata.
The former defines the value of a nondeterministic or alternating automaton on a word to be the supremum/infimum of its runs' values, having the `choice' and `obligation' interpretation of nondeterminism and universality, respectively. (See \cite{Bok21} for a discussion on the differences between the two.)
Since history determinism naturally relates to `choice' and `obligation' in nondeterministic and alternating automata, quantitative automata better fit the present work.

Due to space constraints, some of the proofs appear in the appendix.

\section{Preliminaries}\label{sec:Preliminaries}
\Subject{Words}
An \emph{alphabet} $\Sigma$ is a finite nonempty set of letters. A finite (resp.\ infinite) \emph{word} $u=u_0 \ldots u_k\in \Sigma^{*}$ (resp.\ $w=w_0 w_1\ldots\in \Sigma^{\omega}$) is a finite (resp.\ infinite) sequence of letters from $\Sigma$.  
We write $\Sigma^\infty$ for $\Sigma^* \cup \Sigma^\omega$.
We use $[i..j]$ to denote a set $\{i,\ldots,j\}$ of integers, $[i]$ for $[i..i]$, $[..j]$ for $[0..j]$, and $[i..]$ for integers equal to or larger than $i$. We write $w[i..j], w[..j]$, and $w[i..]$ for the infix $w_i \ldots w_j$, prefix $w_0 \ldots w_j$, and suffix $w_i \ldots$ of $w$.
A \emph{language} is a set of words, and the empty word is written $\epsilon$.

\Subject{Games}

We consider turn-based zero-sum games between Adam and Eve, with $\Sigma$-labelled transitions. A play generates a word, and each word has a value, given by the game's payoff function. % $f:\Sigma^* \to \Reals$ or $f:\Sigma^\omega \to \Reals$. 
Eve tries to maximise the value of the play, while Adam tries to minimise it.
Formally, for a payoff function $f$, an \emph{$f$ game} is defined on an \emph{arena} $(V, E, V_E,V_A, L: E\rightarrow \Sigma\cup \{\varepsilon\})$, which consists of a (potentially infinite) set of positions $V$, partitioned into Eve's positions $V_E$ and Adam's positions $V_A$, and a set of edges $E\subseteq V\times V$, labelled by $L$ with letters from $\Sigma\cup \{\varepsilon\}$. In infinite-duration games every position has at least one outgoing edge. A play is a maximal path over $V$; its non-$\varepsilon$ labels induce a word $w\in \Sigma^{*}$ or $\Sigma^\omega$. The payoff of a play is the value of this word, given by the payoff function $f$. %:\Sigma^* \to \Reals$ or $f:\Sigma^\omega \to \Reals$.

 Strategies for Adam and Eve map partial plays ending in a position $v$ in $V_A$ and $V_E$ respectively to outgoing edges from $v$. A play or partial play $\pi$ agrees with a strategy $\strat_P$, written $\pi\in \strat_P$, for a player $P\in \{A,E\}$, if whenever its prefix $p$ ends in a position in $V_P$, the next edge is $\strat_P(p)$. The value $f(\strat_E)$ of a strategy $\strat_E$ for Eve is $\inf_{\pi\in \strat_E} f(\pi)$ and the value $f(\strat_A)$ of a strategy $\strat_A$ for Adam is $\sup_{\pi\in\strat_A} f(\pi)$. 
 Let $S_E$ and $S_A$ be the sets of strategies for Eve and Adam respectively. If $\sup_{\strat \in S_E} f(\strat)$ (the best Eve can do) coincides with $\inf_{\strat \in S_A} f(\strat)$ (the best Adam can do), we say that $G$ is determined and $\sup_{\strat \in S_E} f(\strat)=\inf_{\strat \in S_A} f(\strat)$ is called the value of $G$.
Eve wins the $t$-threshold game on $G$, for some $t\in\Reals$, if the value of $G$ is at least $t$; else Adam wins. Eve wins the strict $t$-threshold game on $G$ if the value of $G$ is greater than $t$.
Two games are equivalent in this context if they have the same value.
We restrict the scope of this article to determined games.

\Subject{Quantitative Automata}

An \emph{alternating quantitative automaton} on words is a tuple $\A=(\Sigma,Q,\iota,\delta)$, where: $\Sigma$ is an alphabet; $Q$ is a finite nonempty set of states; $\iota\in Q$ is an initial state; and $\delta\colon Q\times \Sigma \to \BP(\Rat \times Q)$ is a transition function, where $\BP(\Rat\times Q)$ is the set of positive Boolean formulas (\emph{transition conditions}) over weight-state pairs. % and $\gamma:Q \times \Sigma \times Q \to \Rat$ is a weight function.

A \emph{transition} is a tuple $(q,\letter,\weight,q')\in Q{\times}\Sigma{\times} \Rat\times Q$, sometimes also written $\trans{q}{\letter:\weight }{q'}$. (Note that there might be several transitions with different weights over the same letter between the same pair of states\footnote{This extra flexibility of allowing for ``parallel'' transitions with different weights is often omitted (e.g., in \cite{CDH09}) since it is redundant for some value functions while important for others.}.) We write $\gamma(t)=\weight$ for the weight of a transition $t=(q,\letter,\weight,q')$.

An automaton $\A$ is nondeterministic (resp.\ universal) if all its transition conditions are disjunctions (resp.\ conjunctions), and it is deterministic if all its transition conditions are just weight-state pairs. We represent the transition function of nondeterministic and universal automata as $\delta\colon Q\times \Sigma\to 2^{(\Rat\times Q)}$, and of a deterministic automaton as $\delta\colon Q\times \Sigma\to \Rat\times Q$. 

We require that the automaton $\A$ is $\emph{total}$, namely that for every state $q\in Q$ and letter $\letter\in\Sigma$, there is at least one state $q'$ and a transition $\trans{q}{\letter:\weight}{q'}$.
For a state $q\in Q$, we denote by $\A^q$ the automaton that is derived from $\A$ by setting its initial state $\iota$ to $q$. 

A run of the automaton on a word $w$ is intuitively a play between Adam and Eve. It starts in the initial state $\iota$, and in each round, when the automaton is in state $q$ and the next letter of $w$ is $\letter$, Eve resolves the nondeterminism (disjunctions) of the transition condition $\delta(q,\letter)$ and Adam resolves its universality (conjunctions), yielding a transition $\trans{q}{\letter:\weight}{q'}$. The output of a play is thus a sequence $\pi = t_0 t_1 t_2 \ldots$ of transitions. As each transition $t_i$ carries a weight $\gamma(t_i)\in\Rat$, the sequence $\pi$ provides a weight sequence $\gamma(\pi) = \gamma(t_0) \gamma(t_1) \gamma(t_2) \ldots$. 
More formally, given the automaton $\A=(\Sigma, Q,\iota,\delta)$ and a word $w\in \Sigma^*$ (resp.\ $w\in \Sigma^\omega$), we define the arena $G(\A,w)$ 
with positions $Q\times\Sigma^* \times \BP(\Rat\times Q)$ (resp.\ $Q\times\Sigma^\omega \times \BP(\Rat\times Q)$), 
the initial position $(\iota, w,\delta(\iota, w[0]))$,
$\varepsilon$-labelled edges from $(q,u,b)$ to $(q,u,b')$ when $b'$ is an immediate subformula of $b$, and $x$-labelled edges from $(q,u,(x,q'))$ to $(q',u[1..],\delta(q',u[1]))$. Conjunctive positions belong to Adam while disjunctive ones belong to Eve.

A $\Val$ automaton (for example a $\Sum$ automaton) is one equipped with a \emph{value function} $\Val:\Rat^* \to \Reals$ or $\Val:\Rat^\omega \to \Reals$. The corresponding game is the $\Val$ game on the arena $G(\A,w)$: each run $\pi$ (play in $G(\A,w)$) has a real value $\Val(\gamma(\pi))$, which we abbreviate by $\Val(\pi)$. 
When this game is determined, we say that the value of $\A(w)$ is the value of $G(\A,w)$, and if $G(\A,w)$ is determined for all $w\in \Sigma^{\omega}$, we say that $\A$ realizes a function from words to real numbers.
We restrict the scope of this article to automata realizing functions.
 
Two automata $\A$ and $\A'$ are \emph{equivalent}, denoted by $\A\equiv\A'$, if they realize the same function.
For a threshold $t\in\Reals$ and a $\Val$ automaton $\A$, we also speak of a corresponding Boolean $t$-threshold $\Val$ automaton $\A'$ that accepts the words $w$ such that $\A(w)\geq t$.

Observe that when $\A$ is nondeterministic, a run of $\A$ on a word $w$ is a sequence $\pi$ of transitions, and the value of $\A$ on $w$ is the supremum of $\Val(\pi)$ over all these runs $\pi$.

\Subject{Value functions}
We list here the most common value functions for quantitative automata on finite/infinite words, defined over sequences of rational weights\footnote{There are also value functions that are more naturally defined over sequences of tuples of rational numbers, for example discounted-summation with multiple discount factors \cite{BH21}.}:
\begin{itemize} 
\item For finite sequences $v=v_0 v_1 \ldots v_{n-1}$:
\vspace{-.3cm}
\begin{multicols}{2}
	\begin{itemize}
	\item $\displaystyle \Sum(v) = \sum_{i=0}^{n-1} v_i$
	\item $\displaystyle \Avg(v) = \frac{1}{n} \sum_{i=0}^{n-1} v_i$
\end{itemize}
\end{multicols} \vspace{-.2cm}
\item For finite and infinite sequences $v=v_0 v_1 \ldots$:
\vspace{-.2cm}
\begin{multicols}{2}
	\begin{itemize}
	\item $\displaystyle \Inf(v) = \inf\{v_n \St n \geq 0\}$
	\item $\displaystyle \Sup(v) = \sup\{v_n \St n \geq 0\}$
\end{itemize}
\end{multicols}
\vspace{-.5cm}
	\begin{itemize}
	\item For a discount factor $\lambda\in\Rat\cap(0,1)$, $\displaystyle \DSum(v) = \sum_{i\geq 0} \lambda^i  v_i$
	%	\item For a series $\Lambda=\lambda_0 \lambda_1 \ldots\in(\Rat\cap(0,1))^\omega$ of discount factors, $\displaystyle \DS(v,\Lambda) = \sum_{i=0}^\infty \Big(v_i \prod_{j=0}^{i-1} \lambda_j\Big)$
\end{itemize}
\vspace{-.5cm}
\item For infinite sequences $v=v_0 v_1 \ldots$:
\vspace{-.2cm}
\begin{multicols}{2}
	\begin{itemize}
	\item $\displaystyle \LimInf(v) = \lim_{n\to\infty}\limits\inf\{v_i \St i \geq n\}$
	\item $\displaystyle \LimSup(v) = \lim_{n\to\infty}\limits\sup\{v_i \St i \geq n\}$
\end{itemize}
\end{multicols}
\vspace{-.7cm}
\begin{itemize}
	\item $\displaystyle \LimInfAvg(v) ~=\, \LimInf (\Avg(v_0), \Avg(v_0,v_1), \Avg(v_0,v_1,v_2),\ldots)$\\[-0.3cm]
	\item $\displaystyle \LimSupAvg(v) = \LimSup (\Avg(v_0), \Avg(v_0,v_1), \Avg(v_0,v_1,v_2),\ldots)$
\end{itemize}
\end{itemize}
\vspace{-.2cm}
($\LimInfAvg$ and $\LimSupAvg$ are also called $\MPL$ and $\MPH$.)

\Subject{Products}
The synchronized product of a $\Sigma$-labelled game $G$ and an automaton $\A$ over alphabet $\Sigma$ is (like in the Boolean setting, see e.g., \cite[Definition 1]{BL19}) a game $G\times \A$ obtained by taking the product of the positions of $G$ and the states and transition conditions of $\A$, and their corresponding transitions. Positions with nondeterminism are of Eve and positions with universality are of Adam.  Transitions carry their weight from the corresponding transition in $\A$. The payoff function of the game is the value function of $\A$. 

%Formally, the synchronised product of a $\Sigma$-labelled game $G=(Q_G,\iota_G,\delta_G)$ and an automaton $\A=(\Sigma,Q,\iota,\delta)$ over $\Sigma$ is the game $(Q_G\times Q,\iota_G\times\iota, \delta')$ where $\delta'(q,p)$ consists of $\delta_G(q)$ where each atom $(x,q')$ is substituted with $\delta(p,x)[(y,(q',p'))/(y, p')]$. In other words, the first layer of the transition function describes the transition function of $G$, producing a successor state $q'$ in $G$ and a letter $x$ from $\Sigma$; then the second layer of the transition function describes the transitions of $\A$ processing $x$, finally producing a successor state $p'$ in $\A$ and a weight $y\in \Rat$.

\section{Good For Gameness, History Determinism, and Determinizability By Pruning}
\label{sec:QuantitativeGfgDefinitions}

In the Boolean setting, ``good for gameness'' and ``history determinism'', stemming from different concepts, coincide both for nondeterministic and alternating automata \cite{BL19}. 

We generalize these definitions to quantitative automata, observing that under this setting they need some sub-variants, relating to whether one considers automata/games equivalence with respect to all values or some threshold\footnote{For a threshold $t\in\Reals$, we provide the definitions with respect to a non-strict inequality $\geq t$ . Using strict inequality $>t$ instead, yields the same relations between the notions, as stated in \cref{cl:NotionRelations}.}. As shown in \cref{sec:comparisons}, the two main notions, as well as some of their variants, are generally not equivalent in the quantitative setting.

The notion of determinizability by pruning, which has an appealing structural definition, is generally stricter than good for gameness and history determinism in the setting of nondeterministic automata, already in the Boolean setting, yet we show that for some value functions it is equivalent to history determinism. For alternating automata, we show that it is incomparable with history determinism and good for gameness.

\begin{definition}[Good for gameness]\label{def:GFGameComposition}
An automaton $\A$ realizing a function $f:\Sigma^* \to \Reals$ or $f:\Sigma^\omega \to \Reals$ is 
\begin{itemize}
\item \emph{good for games} if for every determined\footnote{We discuss in the conclusion  questions that arise if this restriction is lifted} game $G$ with a $\Sigma$-labelled arena and payoff function $f$, we have that  $G$ and $G\times\A$ have the same value; 
\item \emph{good for $t$-threshold games}, for some $t\in\Reals$, if for every determined game $G$ with a $\Sigma$-labelled arena and payoff function $f$, Eve wins the $t$-threshold game on $G$ if and only if she wins the $t$-threshold game on $G\times \A$;
\item \emph{good for threshold games} if it is good for $t$-threshold games for all $t\in\Reals$.
\end{itemize}
\end{definition}

An automaton is history deterministic if there are strategies to resolve its nondeterminism and universality, such that for every word, the (threshold) value remains the same. 
\begin{definition}[History-determinism] \label{def:HistoryDet}
Consider an alternating $\Val$ automaton $\A=(\Sigma,Q,\iota,\delta)$ realizing a function $f:\Sigma^* \to \Reals$ or $f:\Sigma^\omega \to \Reals$. Formally, \emph{history determinism} is defined via \emph{letter games}, detailed below.
\begin{itemize}
	\item $\A$ is \emph{history deterministic} if Eve and Adam win their letter games.
	\item $\A$ is \emph{$t$-threshold history deterministic}, for some $t\in\Reals$, if Eve and Adam win their $t$-threshold letter games.
	\item $\A$ is \emph{threshold history deterministic} if it is $t$-threshold history deterministic for all $t\in\Reals$.
\end{itemize}

Eve's (Adam's) letter games are the following win-lose games, in which Adam (Eve) chooses the next letter and Eve and Adam resolve the nondeterminism and universality, aiming to construct a run whose value is (threshold) equivalent to the generated word's value.
\begin{description}
\item[Eve's letter game:] 
A configuration is a pair $(\sigma, b)$ where $b\in \BP(Q)$ is a transition condition and $\sigma\in\Sigma\cup\{\epsilon\}$ is a letter. (We abuse $\epsilon$ to also be an empty letter.)
A play begins in $(\sigma_0,b_0)=(\epsilon,\iota)$ and consists of an infinite sequence of configurations $(\sigma_0,b_0)(\sigma_1,b_1)\ldots$. 
In a configuration $(\sigma_i,b_i)$, the play proceeds to the next configuration $(\sigma_{i+1},b_{i+1})$ as follows.  
\begin{itemize}
\item If $b_i$ is a state of $Q$, Adam picks a letter $a$ from $\Sigma$, and $(\sigma_{i+1},b_{i+1})=(a,\delta(b_i,a))$.
\item If $b_i$ is a conjunction $b_i=b' \land b''$, Adam chooses between $(\epsilon,b')$ and $(\epsilon,b'')$.
\item If $b_i$ is a disjunction $b_i=b' \lor b''$, Eve chooses between $(\epsilon,b')$ and $(\epsilon,b'')$.
\end{itemize}
In the limit, a play consists of an infinite word $w$ that is derived from the concatenation of $\sigma_0,\sigma_1,\ldots$, as well as an infinite sequence $b_0,b_1,\ldots$ of transition conditions, which yields an infinite sequence $\pi = t_0,t_1,\ldots$ of transitions.

If $\A$ is over infinite words, Eve wins a play in the letter-game  if $\Val(\pi) \geq \A(w)$. In the $t$-threshold letter game, Eve wins if  $\A(w)\geq t \implies \Val(\pi) \geq t$. For $\A$ over finite words, Eve wins if  $\Val(\pi[0..i]) \geq \A(w[0..i])$  or $\A(w[0..i])\geq t \implies \Val(\pi[0..i]) \geq t$ for all $i$.

\item[Adam's letter game] is similar to Eve's game, except that Eve chooses the letters instead of Adam, and Adam wins a play in his letter game if $\Val(\pi) \leq \A(w)$ and in his $t$-threshold letter game if $\A(w)<t\implies \Val(\pi)<\A(w)$. (The asymmetry of $<$ and $\leq$ is intended). 
\end{description}

\end{definition}

Intuitively, an automaton is determinizable by pruning if it can be determinized to an equivalent (w.r.t.\ a threshold) deterministic automaton by removing some of its states and transitions. (In an alternating automaton, ``removing transitions'' means removing some disjunctive and conjunctive choices.)

\begin{definition}[Determinizability by Pruning] \label{def:DBP}
	A $\Val$ automaton $\A$ is
	\begin{itemize}
		\item \emph{determinizable by pruning} if there exists a deterministic $\Val$ automaton $\A'$ that is derived from $\A$ by pruning, such that $\A'\equiv \A$;
		\item \emph{$t$-threshold determinizable by pruning} if there is a deterministic $\Val$ automaton $\A'$ that is derived from $\A$ by pruning, such that for every word $w$, we have $\A'(w) \geq t$ iff $\A(w) \geq t$;
		\item \emph{threshold determinizable by pruning} if it is $t$-threshold determinizable by pruning $\forall t\in\Reals$.
	\end{itemize}
\end{definition}

Observe that a $\Val$-automaton can be good for games, history deterministic, or determinizable by pruning when interpreted on infinite words, but not when interpreted on finite words, as demonstrated in \cref{fig:HdInfiniteNotFinite} .

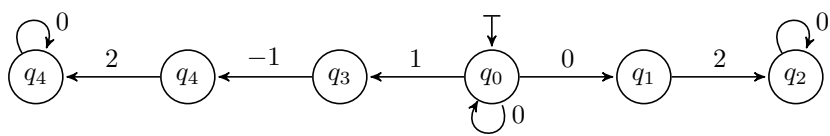
\begin{figure}[h]
	\centering
	\begin{tikzpicture}[->,>=stealth',shorten >=1pt,auto,node distance=2cm, semithick, initial text=, every initial by arrow/.style={|->},state/.style={circle, draw, minimum size=0.5cm}]

		\node[initial above, state] (q0) {$q_0$};
		\node[state] (q1) [ right of=q0,yshift=0.0cm] {$q_1$};
		\node[state] (q2) [right of=q1,] {$q_2$};
		\node[state] (q3) [ left of=q0,yshift=-0.0cm] {$q_3$};
		\node[state] (q4) [left of=q3,xshift=0cm] {$q_4$};
		\node[state] (q5) [left of=q4,xshift=0cm] {$q_4$};
		
		\path 
		(q0) edge	[<-,loop , out=-120, in=-70,looseness=5] node [xshift=0.2cm] {$0$} (q0)
		 		 edge  node {$0$} (q1)		
		 		 edge  node[above]{$1$} (q3)		
		 		 
		(q1) edge  node {$2$} (q2)
		(q2) edge	[loop above, out=120, in=70,looseness=5] node [right,xshift=.2cm]{$0$} (q2)
		
		(q3) edge  node[above]{$-1$} (q4)
		(q4) edge  node[above]{$2$} (q5)
		(q5) edge	[loop above, out=120, in=70,looseness=5] node [right,xshift=.2cm]{$0$} (q5)
		
		;
	\end{tikzpicture}
	\caption{A nondeterministic $\DSum$-automaton with discount factor $\frac{1}{2}$ over a unary alphabet that is determinizable by pruning, good for games, and history deterministic with respect to infinite words, but none of them with respect to finite words: For the single infinite word, the initial choice of going from $q_0$ to $q_1$ provides the optimal value of $1$, making it all of the above. On finite words, on the other hand, it is not even threshold history deterministic (and by \cref{cl:NotionRelations} neither of the rest), since in order to guarantee a value of at least $1$, the first transition should be different  for the word of length $1$ and the word of length $2$, going to $q_3$ for the former and to $q_1$ for the latter.}\label{fig:HdInfiniteNotFinite}
\end{figure}

\section{The Relations Between Notions}\label{sec:comparisons}

\newcommand{\OImplication}{$~\Longleftarrow~$}
\newcommand{\Implication}{$~\Longrightarrow~$}
\newcommand{\NoImplication}{${\hspace{0.15cm}/\hspace{-0.45cm}\Longrightarrow}~$}
\newcommand{\NoImplicationR}{${\hspace{0.15cm}/\hspace{-0.45cm}\Longleftarrow}~$}
\newcommand{\DualImplication}{$~\Longleftrightarrow~$}
\newcommand{\NonDualImplication}{$~^{\Longrightarrow}_{\hspace{0.15cm}/\hspace{-0.30cm}\Longleftarrow}~$}

Having defined these notions, we now establish which inclusions hold in general, and which are conditional on characteristics of the value function, as summarised in~\cref{fig:Comaprison}.

\begin{theorem}\label{cl:NotionRelations}
	(Threshold) good for gameness, (threshold) history determinism, and (threshold) determinizability by pruning of quantitative automata are related as described in \cref{fig:Comaprison}.
\end{theorem}

\begin{figure}[h]
	\centering
	\begin{tikzpicture}[->,>=stealth',shorten >=1pt,auto,node distance=0.75cm, semithick, initial text=, every initial by arrow/.style={|->},state/.style={circle, draw, minimum size=0.5cm}]
		
		\node (GFG) {Good For Gameness {\Large$=$\normalsize$^1$} Threshold Good For Gameness {\Large$\cong$\normalsize$^2$}  Threshold History Determinism};
		\node (M1) [below of = GFG] {} ;
		\node (R1) [left of = M1, xshift = 3.4cm, rotate = 70] {\Large $\subsetneq$} ;
		\node (R1N) [right of = R1, xshift = -0.4cm] {$^4$} ;
		\node  (Rn) [left of = R1, xshift = -0.4cm] {\footnotesize(for nondet.)};
		\node (R2) [right of = M1, xshift = 4.2cm, rotate = 110] {\Large$\subsetneq$} ;
		\node (R2N) [right of = R2, xshift = -0.45cm, yshift=0.1cm] {$^3$} ;
		\node (TDBP) [below of = R1, xshift=-2.5cm] {Threshold Determinizability by Pruning} ;
		\node (HD) [below of = R2, xshift=.6cm] {History Determinism} ;
		\node (R3) [below of = M1, xshift=3.6cm] {\Large$\neq$\normalsize$ ^5$} ;
		\node (M2) [below of = R3] {} ;
		\node (DBP) [below of = M2] {Determinizability by Pruning} ;
		\node (R4) [left of = M2, xshift =-0.5cm, rotate = 110] {\Large $\subsetneq$} ;
		\node (R4N) [right of = R4, xshift = -0.45cm, yshift=0.1cm] {$^3$} ;
		\node (R5) [right of = M2, xshift = 0.5cm, rotate = 70] {\Large $\subsetneq$} ;
		\node (R5N) [right of = R5, xshift = -0.4cm] {$^4$} ;
		\node  (Rn) [right of = R5, xshift = .6cm] {\footnotesize(for nondet.)};
		
		\node(Note1) [below of = DBP, xshift=-8.5cm, yshift = -0.1cm, align=left] {1. Always holds (\cref{cl:GFGiffTGFG}).}	;	
		\node(Note2) at (Note1.west) [anchor=west, yshift=-0.75cm] {2. The \!\!\! \OImplication implication always holds (\cref{cl:THDimpliesTGFG}); 
		The \!\!\! \Implication implication holds at least for};
		\node (Note2a) at (Note2.west) [anchor=west, yshift=-0.5cm] {~~~~all $\Val$ automata whose threshold letter games are determined (\cref{cl:TGFGpartiallyImpliesTHD}),};
		\node (Note2b) at (Note2a.west) [anchor=west, yshift=-0.5cm] {~~~~e.g., for $\Inf,\Sup,\LimInf,\LimSup,\DSum$ and all functions on finite words (\cref{cl:gfg-is-sometimes-thd}).};
		\node  (Note3) at (Note2b.west) [anchor=west, yshift=-0.75cm] {3. Strict containment for all non-trivial value functions with at least three values (\cref{cl:THDnotHD});};
		\node  (Note3a) at (Note3.west) [anchor=west, yshift=-0.5cm] {~~~~Equal (the same notion) for value functions with two values.}; 
		\node  (Note4) at (Note3a.west) [anchor=west, yshift=-0.75cm] {4. Strict containment, in general, for nondeterministic automata (\cref{cl:NondetDbpImpHd,cl:HdNeqDbp}); };
		\node  (Note4a) at (Note4.west) [anchor=west, yshift=-0.5cm] {~~~~Equivalent notions for some nondeterministic $\Val$ automata (\cref{sec:SometimesHdEqDbp});};
		\node  (Note4b) at (Note4a.west) [anchor=west, yshift=-0.5cm] {~~~~Incomparable for alteranting automata (\cref{cl:DbpInCompHd}).};
		\node  (Note5) at (Note4b.west) [anchor=west, yshift=-0.75cm] {5. Incomparable, in general, for value functions with at least three values};
		\node  (Note5a) at (Note5.west) [anchor=west, yshift=-0.5cm] {\,~~~(\cref{cl:THDnotHD,cl:HdNeqDbp,cl:DbpInCompHd});}; 
		\node  (Note5b) at (Note5a.west) [anchor=west, yshift=-0.5cm] {~~~~For value functions with two values, as relation $4$ above.}; 
	\end{tikzpicture}
	\caption{The relations between the different notions.}
	\label{fig:Comaprison}
\end{figure}

	Considering good for gameness, if an automaton $\A$ is good for all games then it is obviously good for all threshold games. The implication for the other direction stems from the fact that every concrete game $G$ has a single value $v$. Then for $G$, it is enough to be good for $v$-threshold games, and for all automata, it is enough to be good for all threshold games.
	
\begin{lemma}\label{cl:GFGiffTGFG}
	Good for Gameness \DualImplication Threshold Good for Gameness.
\end{lemma}

For a $t$-threshold history deterministic automaton $\A$, Eve and Adam have strategies to win their $t$-letter games on $\A$. Thus, whenever Eve or Adam win some $t$-threshold game $G$, they can combine their two winning strategies to win $G \times \A$.
\begin{lemma}\label{cl:THDimpliesTGFG}
Threshold History Determinism \Implication Threshold Good for Gameness
\end{lemma}

For the other direction, we generalize proofs from \cite{BKKS13,BL19}: assuming that the automaton $\A$ is not threshold history deterministic we construct a threshold game $G$ with respect to which $\A$ is not good for composition (namely, the product of $G$ with $\A$ does not have the same winner as $G$). However, to build this game, we assume that either Adam wins Eve's letter game on $\A$ or Eve wins Adam's letter game on $\A$, that is, we assume that the letter games on $\A$ are determined. 
We later show that this determinacy requirement holds for all the specific value functions that we consider in the paper, except for $\LimInfAvg$ and $\LimSupAvg$, on which we leave the determinacy question open.

\begin{lemma}\label{cl:TGFGpartiallyImpliesTHD}
For $\Val$ automata whose threshold letter games are determined, Threshold Good for Gameness \Implication Threshold History Determinism.
\end{lemma}

\begin{proof}
	Consider a $\Val$ automaton $\A$ whose threshold letter games are determined. Then, if $\A$ is not threshold history deterministic, it follows that Adam wins Eve's $t$-letter game on $\A$ for some threshold $t$, or Eve wins Adam's $t$-letter game on $\A$ for some $t$. We show below that in both cases $\A$ is not good for threshold games, proving the contra-positive of the claim.
	
	Assume that Adam wins Eve's $t$-letter game $G_{\A,t}$ on $\A$ for some threshold $t$ with a strategy $s$.
	We can build a one-player $\Sigma$-labelled (infinite) game $G_s$ in which the positions, which all belong to Adam, are the finite words that can be constructed along plays of $G_{\A,t}$ that agree with $s$, and where for every positions $u$ and $u\con a$, there is an $a$-labelled edge from the position $u$ to the position $u\con a$. The empty word $\varepsilon$ is the initial position. In other words, this is the one-player arena in which plays correspond to (infinite) words that occur in the letter game if Adam uses the strategy $s$. Notice that since $s$ is a winning strategy in the $t$-letter game, all words $w$ that are plays of $G_s$ have $\A(w)\geq t$. The $t$-threshold game on $G_s$ is therefore winning for Eve.
	
	We now argue that Adam wins the product game $G_s\times \A$. Indeed, Adam can now use the strategy $s$ to choose directions in $G_s$ according to the run constructed so far in $\A$, and resolve conjunctions in $\A$ according to the history of the word and run so far. Since $s$ is a winning strategy for Adam in the letter game, this guarantees that the resulting run  $\rho$ is such that $\Val(\rho)<t$. Then $\A$ is not threshold-good-for-games, as witnessed by $G_s$.
	
	By a similar argument, if Eve wins Adam's $t$-letter game for some $t$ with a strategy $s$, then we can construct a one-player game $G_s$ in which all positions belong to Eve such that $G_s$ is winning for Adam (i.e., all words have value strictly smaller than $t$), but in the product $G\times \A$, Eve wins, i.e., can force value at least $t$.
	
	Hence if either player has a winning strategy in the other player's threshold letter game for some threshold, then the automaton is not good for threshold games.
\end{proof}

We now show that letter games on $\Val$ automata whose threshold variants define Borel sets are determined. This stems from the fact that their winning condition is a union between two conditions that can be defined by threshold $\Val$ automata or their complement.

\begin{proposition}\label{cl:LetterDeterminedForBorelAutomata}
	If for some value function $\Val$, all threshold $\Val$ automata define Borel sets, then threshold letter games on $\Val$ automata are determined.
\end{proposition}
\begin{proof}
	Consider Eve's $t$-letter game on a $\Val$ automaton $\A$, for some threshold $t\in\Reals$. A play of the game generates a sequence $\rho\in (\Sigma \times V)^\omega$, where $\Sigma$ is $\A$'s alphabet and $V$ is the finite set of its weights. We may view $\rho$ as a pair of sequences $(\rho_\Sigma, \rho_V)$, where $\rho_\Sigma\in\Sigma^\omega$ and $\rho_V\in V^\omega$. Then the winning set of Eve is $\{ \rho \St \Val(\rho_V) \geq t \text{ or }  \A(\rho_\Sigma)< t\}$.

Observe that the set $S_V =  \{ \rho \St \Val(\rho_V)\geq t\}$ can be defined by a $t$-threshold deterministic $\Val$ automaton $\B$, in which the weight of a transition over the input letter $(\sigma, v)$ is $v$. 
Let $\A'$ be a $t$-threshold $\Val$ automaton that is identical to $\A$, except that its alphabet is $\Sigma\times V$, while the transitions are sensitive, as in $\A$, only to the $\Sigma$ component of the input. Then the set $S_\Sigma = \{ \rho \St \A(\rho_\Sigma)\geq t\}$ is defined by $\A'$.

As the winning condition of Eve's letter game is the union of $S_V$ and the complement of $S_\Sigma$, and as both are Borel sets, so is the winning condition. Hence, by \cite{Mar75} the game is determined.

The argument regarding Adam's letter game is analogous.
\end{proof}

A direct corollary of \cref{cl:LetterDeterminedForBorelAutomata} is that for most of the common quantitative automata, we have that good for gameness is equivalent to threshold history determinism. In particular, this is the case for all the concrete value functions that are considered in this paper, except for $\LimInfAvg$ and $\LimSupAvg$, for which we leave this equivalence question open.
(While threshold limit-average functions define Borel sets, and therefore limit-average (mean-payoff) games are determined, we are not aware of a result showing that \emph{nondeterministic} threshold limit-average automata define Borel sets. It is easy to see that deterministic such automata do, but not necessarily nondeterministic ones. For nondeterministic discounted-sum automata, we prove in \cref{cl:gfg-is-sometimes-thd} that they define Borel sets by relying on the continuity of discounted-summation, which does not hold for limit average.)

\begin{theorem}\label{cl:gfg-is-sometimes-thd}
Good For Gameness \DualImplication Threshold History Determinism for all $\Val$ automata on finite words, and $\Inf,\Sup,\LimInf,\LimSup$ and $\DSum$ automata on infinite words.
\end{theorem}
\begin{proof}
	It is enough to show that threshold automata of these types define Borel sets, and then the claim  directly follows from \cref{cl:GFGiffTGFG,cl:THDimpliesTGFG,cl:LetterDeterminedForBorelAutomata}.	
\begin{description}
	\item[Automata on finite words.] Every threshold $\Val$ automaton on finite words defines a set of finite words, which is a countable union of singletons and thus a Borel set.
	
	\item[$\Inf$ and $\Sup$ automata.]  Observe that $\Inf,\Sup$ automata on infinite words are ``almost'' like automata on finite words, in the sense that the value of the automaton on a word is equal to its value on some prefix of the word. 
	Formally, for a $\Sup$ automaton $\A$, we have that the set of infinite words $\{w \in \Sigma^\omega \St \A(w)\geq t\}$ is equal to the set of infinite words $\{ w \in \Sigma^\omega \St \text{ exists } p \in \Nat \text{ such that } \A(w[..p]) \geq t\}$ (when considering $\A$ to operate on finite words). Observe that it is indeed a Borel set, since it is a countable union of open sets. The argument for $\Inf$-automata is analogous, having a countable intersection of closed sets.
	
	\item[$\LimInf$ and $\LimSup$ automata.] Observe that threshold $\LimInf$ and $\LimSup$ automata are equivalent to coB\"uchi and B\"uchi automata, respectively, thus defining $\omega$-regular languages, which are known to be Borel sets \cite{Tho97}.
	
	\item[$\DSum$ automata.]  For $\DSum$ automata the argument stems from the continuity with respect to the Cantor topology of functions defined by $\DSum$ automata.
	
	Consider a $\DSum$ automaton $\A$  and a threshold $t\in\Reals$. 
	Define the following set of infinite words $B_t = \{ w \in \Sigma^\omega \St \text{ for every } n \in \Nat \text{ and } p_0 \in \Nat, \text{ there exists } p>p_0, \text{ such that } \A(w[..p]) \geq t-\frac{1}{n}\}$ (when considering $\A$ to operate on finite words).
	
	Observe that $B_t$ is a Borel set, since $\{w \St \A(w[..p]) \geq t-\frac{1}{n}\}$ is an open set, and the existential and universal quantifiers can be defined by countable unions and intersections.
	
	We claim that $B_t$  is equivalent to the set $A_t = \{w \St \A(w) \geq t\}$, which will prove the required statement. One direction is immediate -- if a word $w$ is in $A_t$ then by the definition of $\A(w)$, there are runs of $\A$ on $w$ whose supremum is at least $t$, admitting the membership of $w$ in $B_t$. 
	
	As for the other direction, we show that for every $n \in \Nat$, there is a run $r$ of $\A$ on $w$, such that $\Val(r)\geq t-\frac{1}{n}$, proving that $w$ is in $A_t$. (One can then even combine these runs to create a single run that attains a value at least $t$.)
	Consider some $n \in \Nat$, and let $R$ be the infinite set of finite runs $r_1, r_2, \ldots$ that witness the membership of $w$ in $B_t$ with respect to $2n$. That is, $r_i$ is a run on a prefix of $w$ of length at least $i$, whose value is at least $t-\frac{1}{2n}$. We create a single run $r$ from $R$ in a ``Konig's lemma'' approach (for simplicity, we detail the construction for a nondeterministic automaton, and later explain how to extend it to an alternating automaton): 
	
	We choose the first transition $t_1$ in $r$ to be a transition that appears as the first transition in infinitely many $r_i$'s. We then choose the next transition $t_2$ to be a transition that appears as the second transition, where $t_1$ is the first transition, in infinitely many $r_i$'s, and so on.
	Notice that $r$ is indeed a run of $\A$ and its value is at least $t-\frac{1}{n}$: By the discounted-sum value function, if the value of a long enough prefix is at least $t-\frac{1}{2n}$, the value of the entire run cannot be smaller than $t-\frac{1}{n}$.
	
	Now, for an alternating automaton, rather than ``choosing transitions'' we need to ``resolve the nondeterminism'', while ensuring that the choice we make appears  in infinitely many runs after the previous nondeterministic and \emph{universal} choices that were already made.
\end{description}
\end{proof}

History determinism and determinizability by pruning obviously imply their threshold versions; \cref{fig:THDnotHD} demonstrates that the converse does not hold.
\begin{lemma}\label{cl:THDnotHD}\
	\begin{itemize}
		\item History Determinism \NonDualImplication Threshold History Determinism; 
		\item Determinizability by Pruning \NonDualImplication Threshold Determinizability by Pruning;
		\item History Determinism \NoImplicationR Threshold Determinizability by Pruning
	\end{itemize}
\end{lemma}
\begin{proof}
	The implications are straightforward: a winning strategy for each player in their letter game is also a winning strategy in their $t$-threshold letter game, for every threshold $t\in \Reals$; further, if an automaton $\A'$ that results from pruning $\A$ is equivalent to $\A$, then for every threshold $t$ and word $w$, if $\A(w) \geq t$ then $\A'(w)\geq t$.
	
	As for the non-implications, \cref{fig:THDnotHD} provides such counter examples, which hold, with some variations, with respect to every non-trivial value function with at least 3 values, and in particular with respect to all value functions discussed in the paper. 
	
	Consider, for example, the automaton $\A$ of \cref{fig:THDnotHD} with respect to the $\Sup$ value function. It is not history deterministic, since if the  nondeterminism in $q_0$ is resolved by going to $q_1$, the resulting automaton is not equivalent to $\A$ with respect to the finite word $aa$ and infinite word $a^\omega$, and if it is resolved by going to $q_2$, the resulting automaton fails on $ab$ and $ab^\omega$.
	
	On the other hand, $\A$ is threshold determinizable by pruning and threshold history deterministic: For a threshold up to $1$, the nondeterminism is resolved by going to $q_1$ and for the threshold $2$ by going to $q_2$.
\end{proof}

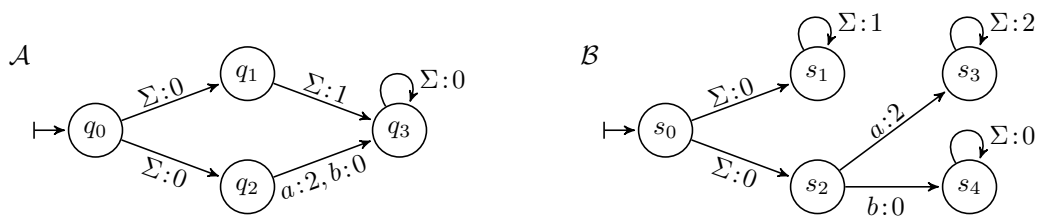
\begin{figure}[h]
	\centering
	\begin{tikzpicture}[->,>=stealth',shorten >=1pt,auto,node distance=2cm, semithick, initial text=, every initial by arrow/.style={|->},state/.style={circle, draw, minimum size=0.5cm}]

		\node (A) {$\A$};
		\node[right of = A,below of =A, initial left, state,xshift=-1cm,yshift=1cm] (q0) {$q_0$};
		\node[state] (q1) [ right of=q0,yshift=0.75cm] {$q_1$};
		\node[state] (q2) [ right of=q0,yshift=-0.75cm] {$q_2$};
		\node[state] (q3) [right of=q0,xshift=2cm] {$q_3$};
		
		\path 
		(q0) edge node[xshift=0.35cm,yshift=0.05cm,rotate=20] {$\LW{\Sigma}{0}$}   (q1)		
		(q0) edge node[below,xshift=0.0cm,yshift=0.05cm,rotate=-20] {$\LW{\Sigma}{0}$}  (q2)		
		
		(q1) edge  node[xshift=-0.45cm,yshift=0.1cm,rotate=-20] {$\LW{\Sigma}{1}$} (q3)
		(q2) edge  node[below,xshift=-0.1cm,yshift=-0.00cm,rotate=20] {$\LW{a}{2}, \LW{b}{0}$}(q3)
		
		(q3) edge	[loop above, out=120, in=70,looseness=5] node [right,xshift=.2cm,yshift=-0.05cm]{$\LW{\Sigma}{0}$} (q3);

		\node[right of = q3,above of =q3,xshift=0.5cm,yshift=-1cm]  (B) {$\B$};
		\node[right of = B,below of =B, initial left, state,xshift=-1cm,yshift=1cm] (s0) {$s_0$};
		\node[state] (s1) [ right of=s0,yshift=0.75cm] {$s_1$};
		\node[state] (s2) [ right of=s0,yshift=-0.75cm] {$s_2$};
		\node[state] (s3) [right of=s1] {$s_3$};
		\node[state] (s4) [right of=s2] {$s_4$};
		
		\path 
		(s0) edge node[xshift=0.35cm,yshift=0.05cm,rotate=20] {$\LW{\Sigma}{0}$}   (s1)		
		(s0) edge node[below,xshift=0.0cm,yshift=0.05cm,rotate=-20] {$\LW{\Sigma}{0}$}  (s2)		
		
		(s1) edge	[loop above, out=120, in=70,looseness=5] node [right,xshift=.2cm,yshift=-0.05cm]{$\LW{\Sigma}{1}$} (s1)
		(s2) edge  node[xshift=0.35cm,yshift=0.2cm,rotate=40] {$\LW{a}{2}$} (s3)
		(s2) edge  node[below,xshift=-0.1cm,yshift=-0.00cm] {$\LW{b}{0}$}(s4)
		
		(s3) edge	[loop above, out=120, in=70,looseness=5] node [right,xshift=.2cm,yshift=-0.05cm]{$\LW{\Sigma}{2}$} (s3)
		(s4) edge	[loop above, out=120, in=70,looseness=5] node [right,xshift=.2cm,yshift=-0.05cm]{$\LW{\Sigma}{0}$} (s4)
		;
		
	\end{tikzpicture}
	\caption{Nondeterministic automata that are threshold history deterministic and threshold determinizable by pruning, but not history deterministic and not determinizable by pruning. The automaton $\A$ has this property with respect, for example, to the $\Sum/\DSum/\Sup$ value functions, and $\B$ with respect, for example, to $\Avg/\LimSup/\LimInf/\LimSupAvg/\LimInfAvg$.}
	\label{fig:THDnotHD}
\end{figure}

\paragraph*{History Determinism $\neq$ Determinizability by Pruning}\label{sec:HdNeqDbp}

{\bf For nondeterministic automata}, it is clear that determinizability by pruning implies history determinism: the pruning provides a strategy for Eve in her letter game.
\begin{proposition}\label{cl:NondetDbpImpHd}
	For nondeterministic automata, determinizability by pruning \Implication history determinism.
\end{proposition}

%The first examples of history deterministic automata embedded deterministic automata \cite{HP06}, and thus were obviously determinizable by pruning, and Colcombet conjectured that history-deterministic alternating automata with $\omega$-regular acceptance conditions are not more concise than deterministic ones \cite{Col12}.
%Yet, this has since been shown to be false: already history deterministic nondeterministic B\"uchi automata cannot be pruned into deterministic ones \cite{BKKS13} and coB\"uchi automata can be exponentially more concise than deterministic ones \cite{KS15}. 

The converse was shown to be false for B\"uchi and coB\"uchi automata \cite{BKKS13}, directly implying the same for $\LimSup$ and $\LimInf$ automata. Considering $\LimInfAvg$ and $\LimSupAvg$ automata, the automaton depicted in \cref{fig:LimAvg}, which is similar to the coB\"uchi automaton in \cite[Figure 3]{BKKS13}, is history deterministic but not determinizable by pruning.

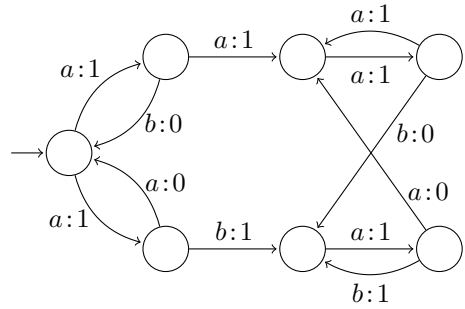
\begin{figure}[htb]
	\begin{center}
		\begin{tikzpicture}[shorten >=1pt,node distance=1.8cm,on grid,auto,initial text=,
			every state/.style={inner sep=0pt,minimum size=6mm},accepting/.style=accepting by arrow]
			\node[state,initial]	(p_0) {};
			\node[state,above right=of p_0] (p_1) {};
			\node[state,below right=of p_0] (p_2) {};
			\node[state,right=of p_1] (q_1) {};
			\node[state,right=of p_2] (q_2) {};
			\node[state,right=of q_1] (r_1){};
			\node[state,right=of q_2] (r_2) {};
			\path[->]
			(p_0) edge[bend left] node[xshift=0.15cm] {$\LW{a}{1}$} (p_1)	
			edge[bend right] node[left] {$\LW{a}{1}$} (p_2)
			(p_1) edge[bend left] node[near start,xshift=-0.2cm] {$\LW{b}{0}$} (p_0)
			edge node {$\LW{a}{1}$} (q_1)
			(p_2) edge[bend right] node[right] {$\LW{a}{0}$} (p_0)
			edge node {$\LW{b}{1}$} (q_2)
			(q_1) edge node[below] {$\LW{a}{1}$} (r_1)
			(q_2) edge node {$\LW{a}{1}$} (r_2)
			(r_1) edge node[near start,xshift=-0.15cm] {$\LW{b}{0}$} (q_2)
			edge[bend right] node[above] {$\LW{a}{1}$} (q_1)
			(r_2) edge node[right,near start]{$\LW{a}{0}$} (q_1)
			edge[bend left] node{$\LW{b}{1}$} (q_2)
			;
		\end{tikzpicture}
	\end{center}
	\vspace{-.5cm}
	\caption{\label{fig:LimAvg} (Similar to \cite[Figure 3]{BKKS13}.) A history deterministic $\LimInfAvg$ or $\LimSupAvg$ automaton that is not determinizable by pruning. (Missing transitions lead to a sink with a $0$-weighted self loop on both $a$ and $b$.)  It is history deterministic by a strategy that chooses in the initial state to go up if and only if it went down the previous time. Following this strategy in the letter game, Eve returns infinitely often to the initial state only on $(aaab)^\omega$, getting a value $\frac{1}{2}$, which is also the automaton's value on it. For every other word, the run of Eve moves to the right part of the automaton, which is deterministic, guaranteeing Eve the optimal value on the word. On the other hand, every pruning of it yields an automaton whose value on either $a^\omega$ or $(ab)^\omega$ is $\frac{1}{2}$ instead of $1$.}
\end{figure}

\begin{proposition}\label{cl:HdNeqDbp}
For nondeterministic $\LimInf$, $\LimSup$, $\LimInfAvg$, and $\LimSupAvg$ automata, history determinism \NoImplication determinizability by pruning.
\end{proposition}

{\bf For alternating automata}, it turns out that (threshold) history determinism and determinizability by pruning are incomparable, as demonstrated in \cref{fig:DbpInCmpHd}.

\begin{proposition}\label{cl:DbpInCompHd}
	For (Boolean and quantitative) alternating automata, determinizability by pruning \NoImplication history determinism, good for gameness.
\end{proposition}

\begin{proof}
	The claim holds for Boolean automata as well as quantitative automata with every non-trivial value function.
	Consider the alternating finite automaton on finite words (which can also be viewed, for example as a  $\Sup$ automaton) in \cref{fig:DbpInCmpHd}. 
	It does not accept any word, and can be determinized by pruning the right nondeterministic transition. However, it is not history deterministic: Eve wins Adam's letter game, by choosing the right nondeterministic transition. 
\end{proof}

\begin{figure}[h]
	\centering
	\begin{tikzpicture}[->,>=stealth',shorten >=1pt,auto,node distance=2cm, semithick, initial text=, every initial by arrow/.style={|->},state/.style={circle, draw, minimum size=0.2cm}]

		\node[initial left, state] (q0) {$q_0$};
		\node[state] (q0a) [ right of=q0,xshift=-0.75cm] {\!\!\small${\lor}$\!\!\!};
		\node[state] (q0b) [ right of=q0a,xshift=-0.95cm, yshift=-0.5cm] {\!\!\small${\land}$\!\!\!};
		\node[state] (q1) [below of=q0,yshift=1.0cm,xshift=0cm] {$q_1$};
		\node[state] (q2) [left of=q1,yshift=0cm,xshift=0cm] {$q_2$};
		\node[state] (q3) [ right of=q0a,yshift=0cm,xshift=0.5cm] {$q_3$};
		\node[state] (q4) [right of=q0a,yshift=-1.0cm,xshift=0.5cm] {$q_4$};
		\node[state] (q5) [right of=q0a,xshift=3cm] {$q_5$};
		\node[state, accepting] (q6) [right of=q4,xshift=0.5cm] {$q_6$};
		
		\path 
		
		(q0) edge  node {$\Sigma$} (q0a)		
		(q0a) edge (q0b)		

		(q0a) edge   (q1)		
		(q1) edge  node [above,rotate=0, xshift=0cm, yshift=0.0cm]{$\Sigma$} (q2)
		
		(q0b) edge   (q3)		
		(q3) edge  node[above,xshift=-0.25cm, yshift=-0.05cm] {$a$} (q5)
		(q3) edge  node[xshift=-0.35cm, yshift=-0.05cm,rotate=20] {$b$} (q6)
		
		(q0b) edge   (q4)		
		(q4) edge  node[xshift=-0.4cm, yshift=-0.25cm,rotate=20] {$b$} (q5)
		(q4) edge  node[above,xshift=-0.55cm, yshift=-0.35cm] {$a$} (q6)

		(q2) edge	[loop left] node {$\Sigma$} (q2)
		(q5) edge	[loop right] node {$\Sigma$} (q5)
		(q6) edge	[loop right] node {$\Sigma$} (q6)
		
		;
	\end{tikzpicture}
	\caption{An alternating finite automaton on finite words that is determinizable by pruning, but not history deterministic nor good for games.}\label{fig:DbpInCmpHd}
\end{figure}
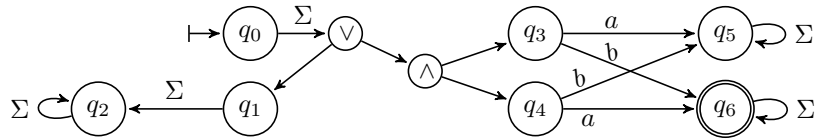

\subsection{When (History Determinism $=$ Determinizability by Pruning)}\label{sec:SometimesHdEqDbp}

In general for nondeterministic automata, determinizability by pruning is strictly contained in history determinism; here we study when the two notions coincide.
In the Boolean setting, they are equivalent for nondeterministic finite automata on finite words (NFAs) \cite{KSV06} as well as for nondeterministic weak automata on infinite words \cite{Mor03}.
Here we analyse general properties of value functions that guarantee this equivalence, and then consider specific value functions on finite and infinite words.
The general properties that we analyze relate to how ``sensitive'' the value function is to the prefix, current position, and suffix of the weight sequence.

We begin by defining \textit{cautious\footnote{Similar transitions are sometimes called ``residual'' in the literature.} strategies} for Eve in 
the letter game, that we then use to define value functions that are ``present focused''.
 Intuitively, a strategy is cautious if it avoids mistakes, that is, it only builds run prefixes that can still achieve the maximal value of any continuation of the word so far. 

\begin{definition}[Cautious strategies]
Consider Eve's letter game on a $\Val$ automaton $\A$.
	A move (transition) $t=q\xrightarrow{\sigma:\weight}q'$ of Eve, played after some run $\rho$ ending in a state $q$, is \emph{non-cautious}
	if for some word $w$, there is a run $\pi'$ from $q$ over $\sigma w$ such that $\Val(\rho\pi')$ is strictly greater than the value of $\Val(\rho\pi)$ for any $\pi$ starting with $t$.
	
	A strategy is \emph{cautious} if it makes no non-cautious moves.
\end{definition}

We call a value function \emph{present focused} if, morally, it depends on the prefixes of the value sequence, formalized by winning the letter game via cautious strategies.

\begin{definition}[Present-focused value functions]
	A value function $\Val$, on finite or infinite sequences, is \emph{present focused} if for all automata $\A$ with value function $\Val$, every cautious strategy in the letter game on $\A$ is also a winning strategy in that game.
\end{definition}

\noindent Value functions on finite sequences are present focused, as they can only depend on prefixes.
\begin{lemma}\label{cl:finite-is-present-oriented}
	Every value function  $\Val$ on finite sequences is present focused.
\end{lemma}
\begin{proof}
	Assume Eve plays a cautious strategy $s$ in some letter game on an automaton $\A$ on finite words. Towards a contradiction, assume that there is a  finite play $\pi$, in which Adam plays some word $w$ and Eve plays a run $\rho$ over $w$ such that $\Val(\rho)<\A(w)$. Then,  let $\rho'$ be the longest prefix of $\rho$ such that the highest value of a run over $w$ starting with $\rho'$ is $\A(w)$. Since $\rho$ is not a run with value $\A(w)$, $\rho'$ is a strict prefix of $\rho$. However, since $\rho'$ is the longest prefix that could be continued into a run with value $\A(w)$, Eve's next move after $\rho'$ must be non-cautious, contradicting that $s$ never plays non-cautious moves.
\end{proof}

\begin{remark}
	Value functions on infinite sequences are not necessarily present focused.
	For example, consider the automaton depicted in \cref{fig:HdInfiniteNotFinite}, but viewed as a 
	$\Sup$ automaton on infinite words rather than a $\DSum$ automaton. Observe that Eve can forever stay in $q_0$, always having the potential to continue to an optimal run with value $2$, but never fulfilling this potential.
	\end{remark}

We now define ``suffix monotonicity'' of value functions, which, with present-focus, will guarantee the equivalence of history determinism and determinizability by pruning.

\begin{definition}
	A value function $\Val$ is \emph{suffix monotonic} if for every finite set $S\subset \Rat$, sequence $\alpha\in S^*$ and sequences $\beta,\beta'\in S^\infty$, we have $\Val(\beta)\geq\Val(\beta')$ iff $\Val(\alpha\beta)\geq\Val(\alpha\beta')$.
\end{definition} 
Observe that the above definition does not consider arbitrary sequences of rational numbers, but rather sequences of finitely many different rational numbers, which is the case in sequences of weights that are generated by runs of quantitative automata.

Value functions that are \emph{suffix dependent} (namely $\Val$ functions such that for every finite set $S\subset \Rat$, sequences $\alpha,\alpha'\in S^*$ and sequence $\beta\in S^\infty \setminus \{\epsilon\}$, we have $\Val(\alpha\beta)= \Val(\alpha'\beta)$) are obviously suffix monotonic. Examples for such value functions are the  acceptance condition of NFAs (i.e, a ``last'' value function, that depends only on the last weight of $0$ for rejection and $1$ for acceptance), all $\omega$-regular conditions (which depend on the states/transitions that are visited infinitely often), $\LimInf$, $\LimSup$, $\LimInfAvg$, and $\LimSupAvg$. 
Examples for value functions that are suffix monotonic but not suffix dependent are $\Sum, \Avg$ and $\DSum$, and examples for value functions that are not suffix monotonic are $\Inf$ and $\Sup$.

We next show that suffix monotonicity together with present-focus guarantee the equivalence of history determinism and determinizability by pruning.
The idea is that under these conditions, every cautious strategy in the letter game can be arbitrarily pruned into a positional strategy (with respect to the automaton states).

\begin{theorem}\label{cl:PresentOrientedAndSuffixMonotonic}
	For nondeterministic $\Val$ automata, where $\Val$ is a present-focused and suffix-monotonic value function, we have that history determinism \DualImplication determinizability by pruning.
\end{theorem}
\begin{proof}
	We show that Eve wins her letter game on $\A$ with a positional strategy, which implies that $\A$ is determinizable by pruning.
	
	Let $s$ be a cautious strategy for Eve in the letter game on $\A$. Let $\hat{s}$ be an arbitrary positional strategy that only uses transitions also used by $s$. We argue that $\hat{s}$ is also cautious.
	Indeed, if $\hat{s}$ chooses $\tau=q\xrightarrow{\sigma:\weight}q'$ after a play $(\hat{w},\hat{\rho})$ of the letter game, there is some play $(w,\rho)$ from which $s$ plays $\tau$.
	Since $s$ is cautious, for every word $v$ and every run $\pi'$ from $q$ over $\sigma v$, there is a run $\pi$ from $q$ starting with $\tau$ such that  $\Val(\rho\pi)\geq\Val(\rho\pi')$.
 Thus, by suffix monotonicity, we have $\Val(\pi)\geq \Val(\pi')$, and then again by the other direction of suffix monotonicity, we get that $\Val(\hat{\rho}\pi)\geq \Val(\hat{\rho}\pi')$, implying that $\hat{s}$ choosing $\tau$ is a cautious move.
	
	Then $\A$ is determinisable by pruning: the subautomaton $\A_{\hat{s}}$ that only has transitions used by $\hat{s}$ is equivalent to $\A$. Indeed, for every word $w$,
	$\A_{\hat{s}}(w)$ is $\Val(\rho_w)$, where $\rho_w$ is the unique run of $\A_{\hat{s}}$ over $w$. The run $\rho_w$ is also the run built by $\hat{s}$ in the letter game over $w$. Since $\hat{s}$ is cautious and $\Val$ is present focused, we have that $\hat{s}$ is a history-deterministic strategy, which guarantees that $\Val(\rho_w)=\A(w)$, giving us the equivalence of $\A$ and $\A_{\hat{s}}$.
\end{proof}

\begin{remark}
	Both present-focus and suffix-monotonicity are necessary in \cref{cl:PresentOrientedAndSuffixMonotonic}. For example $\LimInf$ is suffix monotonic, but $\LimInf$ automata are not determinizable by pruning. On the other hand, \cref{fig:HdPresentOrientedlNotDBP} demonstrates a present-focused value function whose history deterministic automata on finite words are not determinizable by pruning.
\end{remark}

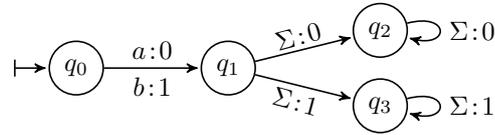
\begin{figure}[h]
	\centering
	\begin{tikzpicture}[->,>=stealth',shorten >=1pt,auto,node distance=2cm, semithick, initial text=, every initial by arrow/.style={|->},state/.style={circle, draw, minimum size=0.5cm}]

		\node[initial left, state] (q0) {$q_0$};
		\node[state] (q1) [right of=q0] {$q_1$};
		\node[state] (q2) [ right of=q1,yshift=0.5cm] {$q_2$};
		\node[state] (q3) [ right of=q1,yshift=-0.5cm] {$q_3$};
		
		\path 
		(q0) edge  node {$\LW{a}{0}$}  node[below] {$\LW{b}{1}$} (q1)		
		(q1) edge  node [xshift=.4cm, yshift=0.05cm,rotate=15]{$\LW{\Sigma}{0}$} (q2)
		(q1) edge  node [left,xshift=.3cm, yshift=-.3cm,rotate=-15]{$\LW{\Sigma}{1}$} (q3)
		
		(q2) edge	[loop right] node [right]{$\LW{\Sigma}{0}$} (q2)
		(q3) edge	[loop right] node [right]{$\LW{\Sigma}{1}$} (q3)
		
		;
	\end{tikzpicture}
	\caption{A nondeterministic $\Val$ automaton $\A$ on finite words with the value function $\Val(\rho)=1$ if $\rho$ has both even and odd values, and $0$ otherwise. Notice that $\Val$ is present focused and $\A$ is history deterministic but not determinizable by pruning.}\label{fig:HdPresentOrientedlNotDBP}
\end{figure}

We now apply these results to specific value functions.

\begin{theorem}\label{cl:SumHdDBP}
\footnote{A slightly weaker result is given in \cite[Theorem 5.1]{AKL10}: a $\Sum$ automaton is history deterministic \emph{with a finite-memory strategy for resolving the nondeterminism} if and only if it is determinizable by pruning.}
For nondeterministic $\Sum$ and $\Avg$ automata (on finite words), history determinism \DualImplication determinizability by pruning.
\end{theorem}
\begin{proof}
From \cref{cl:finite-is-present-oriented,cl:PresentOrientedAndSuffixMonotonic} and the suffix monotonicity of these value functions.
\end{proof}

We continue with showing that $\DSum$ is present focused due to the function's continuity.

\begin{lemma}\label{cl:DSum-present-oriented}
	$\DSum$ on infinite sequences is a present-focused value function.
\end{lemma}

\begin{theorem}[{\cite[Section 5]{HPR16}}]
	For nondeterministic $\DSum$ automata on finite and infinite words, history determinism \DualImplication determinizability by pruning.
\end{theorem}
\begin{proof}
	The claim, which was also proved in \cite[Section 5]{HPR16}, is a direct consequence of \cref{cl:finite-is-present-oriented,cl:DSum-present-oriented,cl:PresentOrientedAndSuffixMonotonic} and the suffix monotonicity of the $\DSum$ value functions.
\end{proof}
	
The $\Inf$ and $\Sup$ value function are not suffix monotonic, and indeed the proof of \cref{cl:PresentOrientedAndSuffixMonotonic} does not hold for them -- not every cautious transition of a history deterministic $\Sup$ automaton on finite words can be used for pruning it into a deterministic automaton. Yet, also for $\Inf$ and $\Sup$ automata on finite words we have that history determinism is equivalent to determinizability by pruning, using other characteristics of these value functions -- we can prune the automaton, by choosing the transitions that are used by the strategy of the letter game after reading words with minimal values for $\Sup$ and maximal value for $\Inf$.

\begin{theorem}\label{cl:InfSupHdImpDbp}
	For nondeterministic $\Inf$ and $\Sup$ automata on finite words, history determinism \DualImplication ~determinizability by pruning.
\end{theorem}

%\section{Deciding history determinism}\label{sec:deciding}

%\subsection{LimSup via two tokens}
%\input{decidingGFG.tex}

%\subsection{Finite words, DSum, Sum and Sup via one token}
%\input{onetoken.tex}

\section{Applications to Quantitative Synthesis}\label{sec:synthesis}
\newcommand{\Input}{\item[Input:]}
\newcommand{\System}{\item[System:]}
\newcommand{\Realisability}{\item[Realisability:]}
\newcommand{\Synthesis}{\item[Synthesis:]}

\newcommand{\Part}[1]{\paragraph*{#1}}

Establishing the non-equivalence of history determinism, good for gameness and their threshold versions leaves us with the question of which definitions, if any, are the most useful or interesting ones. 
We explore this question from the perspective of quantitative synthesis. %, and show that each of the notions is closely related to a variant of quantitative synthesis.

In the Boolean setting, Church's classical synthesis problem asks for a transducer $\T$ that produces, letter by letter, for every input sequence $I\in\Sigma_I^\omega$
an output sequence $\T(I)\in\Sigma_O^\omega$ such that $I\otimes \T(I)\in L$ for some specification language $L\in(\Sigma_I\otimes\Sigma_O)^\omega$. %(``Letter by letter'' means that $\T$ implements a function from $\Sigma_I^*$ to $\Sigma_O$, producing the next output letter according to the input prefix read so far.)
This synthesis requirement is \emph{global}, in the sense that the output of all input sequences should satisfy the same constraint. 
A \emph{local} variant of the problem, termed ``good enough synthesis'' in \cite{AK20}, considers each input sequence $I$ separately, requiring that the output $\T(I)$ of the transducer on the input $I$ satisfies $I\otimes \T(I)\in L$ only if $I\otimes O\in L$ for some sequence $O\in\Sigma_O^\omega$. 

In quantitative synthesis, the specification is a function $f: (\Sigma_I \times \Sigma_O)^\omega \rightarrow \Reals$ (generalizing languages $L: (\Sigma_I \times \Sigma_O)^\omega \rightarrow \{\True,\False\}$ ), and the two synthesis problems above naturally generalize into two quantitative variants each -- requiring either the best possible value or a value matching a given threshold.
We thus have four variants of quantitative synthesis: Global/Local Threshold/Best-value synthesis.
It turns out that {\bf good for gameness is closely related to global synthesis, while history determinism is closely related to local synthesis}, both for the threshold and best-value settings. % The threshold/best-value variants of the synthesis problem then relate to whether  good for gameness or history determinism is with respect to a given threshold or not. 
%We now detail these synthesis variants and their connection to good for gameness and history determinism. %For simplicity, we focus on nondeterministic rather than alternating automata.

\Subject{Global Threshold and Best-value Synthesis}
The global threshold variant is the closest to Church synthesis: given a function $f$ and a threshold $t\in\Reals$, it requires that $f(I\otimes \T(I)) \geq t$ for all input sequences $I$.
In the best-value version, $t$ is not given and we are interested in what is the highest threshold that the system can guarantee. %That is, $t = \sup\{v\in\Reals \St \text{ exists a transducer } g:\Sigma_I^* \to \Sigma_O, \text{ such that for every } I\in\Sigma_I^\omega, f(I,g(I)) \geq v\}$. %(Observe that $g$ is assumed to be a transducer and not and arbitrary function, and that a maximal such value $t$, namely not only a supremum value, is assumed to exist.)
%\kar{Check max vs sup w.r.t. the value of game definition}

Analogously to the Boolean setting, a $t$-threshold good for games $\Val$ automaton $\A$ realizing $f$ can be used instead of a deterministic automaton to solve the global threshold synthesis problem: $\A$ is turned into a $t$-threshold $\Val$ game $G_\A$, in which Adam controls the input letters and Eve controls the output letters. Then, the synthesis problem is realizable if and only if Eve has a winning strategy in $G_\A$. If $\A$ is nondeterministic, Eve's winning strategy in $G_\A$  induces a transducer for the synthesis problem. In the best-value case, the same is true, but $\A$ must be good for games, rather than just for $t$-threshold games, and it is Eve's optimal strategy, if it exists, that induces the solution transducer.

%As above, a good for games $\Val$ automaton $\A$ realizing $f$ can be used instead of a deterministic automaton to solve the global best-value synthesis problem: $\A$ is turned into a $\Val$ game $G_\A$, in which Adam controls the input letters and Eve controls the output letters. If $\A$ is nondeterministic, an optimal strategy for Eve in $G_\A$ induces a transducer for the global best-value synthesis problem. If $\A$ is alternating, the value of $G_\A$ corresponds to the best achievable value for the synthesis problem, but a strategy for Eve does not directly induce a transducer.

\Subject{Local Best-value and Threshold Synthesis}
 We define $\best{f}{I} = \sup_{O\in\Sigma_O^\omega} f(I\otimes O)$ for $I\in\Sigma_I^\omega$, i.e., the best value that the input $I$ can get, or converge to, according to $f$.
The local best value synthesis problem requires that for every $I\in\Sigma_I^\omega$, we have $f(I\otimes \T(I)) = \best{f}{I}$.
Since $\best{f}{I}$ is a supremum, it need not be attained by any word; then the synthesis problem is unrealisable, even if the system could force a value arbitrarily close to $\best{f}{I}$. %Approximative synthesis is a version of the synthesis problem that requires the system to guarantee some bound on the distance between the best value and the achieved value.
The threshold variant requires that for every $I\in\Sigma_I^\omega$, such that $\best{f}{I} \geq t$, we have $f(I\otimes \T(I)) \geq t$, for a given threshold $t\in\Reals$.

%\todo{Refer to synthesis without regret; see how this reference goes along the over all relation to the literature.}

The local best value (or $t$-threshold) synthesis problem of a function given by deterministic (or even history-deterministic nondeterministic) automata and the problem of whether a nondeterministic automaton is ($t$-threshold) history deterministic reduce to each other. 
%Similarly, solving local synthesis with a threshold $t$ and deciding $t$-threshold history determinism reduce to each other.
 The relationship between good-enough synthesis~\cite{AK20} and history determinism was noted for visibly pushdown automata in~\cite{GJLZ21}; a similar reduction in~\cite{FLW20} reduces the approximative local best-value synthesis of deterministic quantitative automata over finite words by finite transducers to the notion of $r$-regret determinisability, that is,  whether a nondeterministic automaton is close enough to a deterministic automaton obtained by pruning its product with a finite memory. Our reductions are in the same spirit, but relate the synthesis problem to history determinism rather than determinisability, and obtain a two-way correspondence for all history-deterministic nondeterministic quantitative automata. In the alternating case, only one direction is preserved, and only for realisability, rather than synthesis.

%and in the context of regret-minimisation in~\cite{HPR15}.
%\todo{Do we need to cite some paper for the result or parts of it?}
\begin{proposition}\label{cl:BestValueSynthesisAndHistoryDeterminism}
	Deciding the local best value (resp. $t$-threshold) synthesis problem with respect to a function $f$ given by a ($t$-threshold) history deterministic nondeterministic $\Val$-automaton $\A$ and deciding whether a nondeterministic $\Val$-automaton $\A'$ is ($t$-threshold) history deterministic are linearly inter-reducible. 
	Furthermore, the witness of ($t$-threshold) history determinism of $\A'$ is implementable by the same computational models as a solution to the best-value ($t$-threshold) synthesis of $\A$.
\end{proposition}

%In brief, a witness of ($t$-threshold) history-determinism for a nondeterministic automaton mapping $I\in \Sigma^\omega$ onto $\best{\A}{I}$ constructed from $\A$ is exactly a solution to the local ($t$-threshold) best value synthesis problem. Conversely, if $\A$ is a nondeterministic automaton, a solution to the synthesis problem for a deterministic automaton that maps runs of $\A$, encoded as input-transition pairs, to their value, is exactly a witness of the history-determinism of $\A$.

\NotNeeded{
\Subject{Local Threshold Synthesis}
This synthesis problem is similar to local best value synthesis, just that it also gets a threshold $t\in\Reals$ and requires that for every $I\in\Sigma_I^\omega$, such that $\best{f}{I} \geq t$, we have $f(I\otimes \T(I)) \geq t$.

\begin{proposition}\label{cl:ThresholsSynthesisAndThresholdHistoryDeterminism}
	Deciding the local $t$-threshold  synthesis problem with respect to a function $f$ that is realized by a history deterministic nondeterministic $\Val$-automaton $\A$ 
	and deciding whether a nondeterministic $\Val$-automaton $\A'$ is $t$-threshold history-deterministic are linearly inter-reducible. 
	Furthermore, the witness of history determinism of $\A'$ is implementable by the same computational models as a solution to the best-value synthesis of $\A$.
\end{proposition}
}

\section{Conclusions}
We have painted a picture of how definitions of good for gameness and history determinism  behave in the quantitative setting, and how they relate to quantitative synthesis. Our work opens up many directions for further work, of which we name a few.
\begin{itemize}
\item The reductions between local synthesis and history determinism motivate expanding methods used to decide history determinism of $\omega$-regular automata to quantitative ones.
\item So far, we have restricted our attention to determined games, but one could also consider more general classes of games and study the effect of composition in that setting.
\item One appeal of good for games and history deterministic automata is that they can be more expressive and more succinct than deterministic ones, while their synthesis problems retain the same complexity. The expressivity and succinctness of quantitative good for games and history deterministic automata is open for most value functions.
\item It is natural to look at approximative versions of the discussed notions (as has been done, see the related work section); we expect our results to also generalise in that direction.
\end{itemize}

\newpage
\bibliography{gfg}

%\newpage
\noindent{\huge{\textbf{Appendix}}}
\appendix

\section{Proofs of Section~\ref{sec:comparisons}}\label{ap:comparisons}

\begin{proof}[Proof of \cref{cl:GFGiffTGFG}]
	One direction is immediate: if an automaton $\A$ is good for all games then it is also good for all threshold games. Indeed, assuming that $\A$ is good for games, if the value of a game $G$ is $v$, then the value of the product game $G\times \A$ is also $v$. Then, for all thresholds $t$, Eve wins the $t$-threshold game on both $G$ and $G\times \A$ if and only if $v\geq t$.
	
	As for the other direction, assume $\A$ is good for threshold games. Let $G$ be a game with value $v$. Since $\A$ composes with threshold games, considering the $v$-threshold game on $G$, we know that Eve can achieve at least $v$ in the product $v$-threshold game $G\times A$. Conversely, let $v'\geq v$ be the value of $G\times A$. Since Eve wins the $v'$-threshold game on $G\times A$, and $\A$ is good for threshold games, Eve can also achieve at least $v'$ in $G$, i.e., $v'=v$, the value of $G$.
\end{proof}

\begin{proof}[Proof of \cref{cl:THDimpliesTGFG}]
	Consider a threshold history deterministic automaton $\A$ over an alphabet $\Sigma$, realizing a function $f$. 
	Then for every threshold $t\in\Reals$, Eve has a winning strategy $s'$ in the $t$-threshold letter game on $\A$.
	
	Now, consider a $\Sigma$-labelled $t$-threshold game $G$ with payoff function $f$, in which Eve has a winning strategy $s$.
	Then in the product game $G\times \A$, Eve can combine $s$ and $s'$  into a strategy $\hat s$, so that $s$ guarantees that any play $\pi=(w,\rho)$ that agrees with $\hat s$ reads a word $w$ such that $\A(w)\geq t$, and $s'$ guarantees that $\Val(\rho)\geq t$ (since $\A(w)\geq t$). 
	
	By a similar argument, if Adam has a winning strategy in his threshold letter game, he can combine it with his winning strategy in a threshold game for getting a winning strategy in the product threshold game.
\end{proof}

\subsection{Proofs of Section~\ref{sec:SometimesHdEqDbp}}\label{ap:SometimesHdEqDbp}

\begin{proof}[Proof of \cref{cl:DSum-present-oriented}]
	Consider a $\lambda$-DSum $\Val$ automaton $\A$ and let $m$ be the maximal absolute transition weight in $\A$. Observe that for every word $w$ and state $q$ of $\A$, we have $|\A^q(w)| \leq \frac{m}{1-\lambda}$. 
	%Hence, for every run of $\A$ and position $i$ of it, its suffix $\pi$ from position $i$ has absolute value at most $\frac{m\cdot \lambda^i}{1-\lambda}$.
	
	Let $s$ be a cautious strategy of Eve in the letter game on $\A$.
	By the definition of a cautious strategy, for every finite word $u$, playing according to $s$ on $u$ generates a finite run $\rho$ that ends in some state $q$, such that for every infinite word $v$, there is an infinite run $\pi$ on $v$ from $q$, such that $\Val(\rho \pi) = \A(uv)$.
	
	Now, consider a word $w$, let $r$ be the run of $\A$ on $w$ that is generated by following $s$, and let $r'$ be an optimal run of $\A$ on $w$. 
	For every position $i$, let $q_i$ be the state that $r[0..i]$ ends in and $q'_i$ be the state that $r'[0..i]$ ends in.
	By the cautiousness of $s$, for every position $i$, there is a run $\pi$ from  $q_i$ on $w[i+1..] $, such that for every run $\pi'$ from  $q'_i$ on $w[i+1..] $, we have $\Val(r[0..i]) + \lambda^i\Val(\pi) \geq \Val(r'[0..i]) + \lambda^i\Val(\pi')$.
	
	Since $\Val(\pi)\leq \frac{m}{1-\lambda}$ and $\Val(\pi')\geq -\frac{m}{1-\lambda}$, we get that  $\Val(r'[0..i]) - \Val(r[0..i]) \leq  \frac{2m\cdot \lambda^i}{1-\lambda}$. Since $\lim_{i\to\infty}\frac{2m\cdot \lambda^i}{1-\lambda} = 0$, we get that $\Val(r)=\Val(r')$, implying that Eve wins the letter game.
\end{proof}

\begin{proof}[Proof of \cref{cl:InfSupHdImpDbp}]
	We provide the proof for $\Sup$ automata and then describe the required changes for adapting it to $\Inf$ automata.
	
	Consider a history deterministic $\Sup$ automaton $\A$ on finite words in $\Sigma^*$, whose history determinism is witnessed by a strategy $s$. 
	We derive from $s$ a positional strategy $s'$, by taking for every state $q$ of $\A$ and letter $\sigma\in\Sigma$, the transition that $s$ chooses over a minimal prefix, where minimality is with respect to the $\Sup$ function. 
	
	Formally, for every state $q$, let $m(q)$ be a $\Sup$-minimal run that reaches $q$ along $s$; namely $m(q) = \rho$, such that $\rho$ is a run of $\A$ that agrees with $s$ and ends in $q$, and such that for every run $\rho'$ of $\A$ that agrees with $s$ and ends in $q$, we have $\Sup(\rho) \leq \Sup(\rho')$. (Notice that since there are finitely many weights in $\A$, such a minimal run, which need not be unique, always exists.)  For every state $q$ of $\A$ and letter $\sigma\in\Sigma$, we define $s'(q,\sigma) = t$, such that $s$ chooses $t$ over the prefix run $m(q)$ and current letter $\sigma$.
	
	We claim that $s'$ is cautious. Indeed, for the correctness proof, we shall change $s$ into $s'$ iteratively, considering in each iteration a single state $q$ and letter $\sigma$. Assume by way of contradiction that exists a word $u\in\Sigma^*$ on which $s'$ generates a path $\tau$ that ends in a state $q$, such that $s'(u\sigma) = t$ for a non-cautious transition $t$. Without loss of generality, we may assume that this is not the case for any strict prefix of $u$, as otherwise we can consider that prefix instead of $u$.
	
	By the definition of non-cautiousness, there exists a word $w$, such that the maximal value of $\Sup(\tau\pi)$ for a run $\pi$ from $q$ over $\sigma w$ starting with $t$ is strictly smaller than the maximal value of $\Sup(\tau\pi')$ where $\pi'$ is a run from $q$ over $\sigma w$  that does not start with $t$.
	
	It thus follows that $\Sup(\pi') > \Sup(\tau)$ and that for every run $\pi$ from $q$ over $\sigma w$ starting with $t$, we have $\Sup(\pi') > \Sup(\pi)$. Now, let $\rho$ be a run that witnesses $t$'s minimality in the definition of $s'$, namely $s$ chooses $t$ when reading $\sigma$ after reaching $q$ over $\rho$, and for every run $\rho'$ that ends in $q$, we have $\Sup(\rho) \leq \Sup(\rho')$.
	
	Then, in  particular, $\Sup(\rho) \leq \Sup(\tau)$. Hence, $\Sup(\pi') > \Sup(\rho)$. Therefore, for every run $\pi$ from $q$ over $\sigma w$ starting with $t$, we have $\Sup(\rho\pi') > \Sup(\rho\pi)$, contradicting the cautiousness of $s$.
	
	Having that $s'$ is cautious, we get from \cref{cl:finite-is-present-oriented} that it is also winning in the letter game, implying that the deterministic automaton that results from pruning $\A$ along $s'$ is indeed equivalent to $\A$.
	
	Now, for $\Inf$ automata, the proof is analogous, choosing the $\Inf$-maximal run rather than the $\Sup$-minimal run, switching between some $\geq$ and $\leq$ and between some $<$ and $>$, and providing the following final argument:
	For every run $\pi$ from $q$ over $\sigma w$ starting with $t$, we have $\Inf(\pi) < \Inf(\tau)$ and $\Inf(\pi) < \Inf(\pi')$. Now, let $\rho$ be a run that witnesses $t$'s maximality in the definition of $s'$, namely $s$ chooses $t$ when reading $\sigma$ after reaching $q$ over $\rho$, and for every run $\rho'$ that ends in $q$, we have $\Inf(\rho) \geq \Inf(\rho')$.
	
	Then, in  particular, $\Inf(\rho) \geq \Inf(\tau)$. Hence, for every run $\pi$ from $q$ over $\sigma w$ starting with $t$, we have $\Inf(\pi) < \Inf(\rho)$ and $\Inf(\pi) < \Inf(\pi')$. Hence,  for every run $\pi$ from $q$ over $\sigma w$ starting with $t$, we have $\Inf(\rho\pi) < \Inf(\rho\pi')$, contradicting the cautiousness of $s$.
	
\end{proof}

\section{Proofs of Section~\ref{sec:synthesis}}\label{ap:synthesis}

\begin{proof}[Proof of \cref{cl:BestValueSynthesisAndHistoryDeterminism}]\

	$\Longrightarrow:$ 	Reducing the synthesis problem to the history-determinism problem.
	
		The idea of the reduction (both in the best-value and $t$-threshold case) is to turn output letter choices 
	%expected from the solution to the synthesis problem of a deterministic automaton 
	in $\A$ into nondeterministic choices in $\A'$. Then $\A'$ %is the projection of $\A$ onto the first component of the input letters, which 
	maps $I\in \Sigma_I^\omega$ onto $\best{\A}{I}$. A solution to the synthesis problem for $\A$ corresponds exactly to a function that resolves the nondeterminism of $\A'$ on the fly to build a run with value $\best{\A}{I}$, that is, a witnesses of the history determinism of $\A'$. If $\A$ is itself nondeterministic, then $\A'$ will have both the nondeterminism of $\A$ and the nondeterminism that stems from the choice of output letters. As long as the nondeterminism of $\A$ is history deterministic, the nondeterminism of $\A'$ is history deterministic if and only if $\A$ is local best value realisable.
	
	More formally, first let us define formally the projection of $\A$ onto its first component: $\A'=(\Sigma_I,Q,\iota,\delta')$, where $\delta'(q,a)=\bigvee_{b\in \Sigma_O} \delta(q,(a,b))$.
	In other words, the automaton $\A'$ moves the $\Sigma_O$ letters from the input word into a nondeterministic choice. It implements a mapping of inputs $I\in \Sigma_I^\omega$ onto $\best{\A}{I}$. We now argue that witnesses of history determinism for $\A'$ coincide with solutions to the best-value synthesis problem for $\A$. Let $s$ be the witness of the history determinism of $\A$.

	We first argue that a solution $s'$ to the best-value synthesis problem for $\A$, combined with $s$ is a witness that $\A'$ is history-deterministic. Indeed, in Eve's letter game on $\A'$, Eve has two types of choices: a choice $\bigvee_{b\in \Sigma_O} \delta(q,(a,b))$ of an $\Sigma_O$-letter, and the choice in $\delta(q,(a,b))$ that stems from $\A$. Let $\hat{s}$ be the strategy that after a run prefix $\rho$ ending in a state $q$ over a word $w\in \Sigma_I$ chooses the letter $s'(w)$, that is, the disjunct $\delta(q,(a,s'(w)))$ in the disjunction $\bigvee_{b\in \Sigma_O} \delta(q,(a,b))$. Then, from $\delta(q,(a,s'(w)))$, $\hat{s}$ behaves as $s$ would after a run of $\A$ over $w\otimes\out{s}(w)$.
	
	First, observe that a run $\rho$ of $\A'$ over $I\in \Sigma_I^\omega$, labelled with the choices of $\Sigma_O$-letters forming some $O\in \Sigma_O^\omega$, corresponds to a run of $\A$ over $I\otimes O$ with the same value.
	
	Then, since $s'$ is a solution to the best value synthesis problem, it guarantees that given an input word $\Sigma_I$, the sequence of $\Sigma_O$ letters chosen by $\hat{s}$ is $\out{s}(I)$, and $\A(I\otimes \out{s}(I))=\best{\A}{I}$. Then, as $s$ witnesses the history determinism of $\A$, $\hat{s}$ guarantees that $\rho$ has value $\A(I\otimes \out{s}(I)$, that is, $\hat{s}$ witnesses the history determinism of $\A'$.

	For the converse direction, assume $\A'$ is history deterministic, as witnessed by some strategy $s$. We claim that $s$ induces a solution $s'$ to the synthesis problem for $\A$ as follows: after reading an finite sequence of inputs $Ia\in \Sigma_I^*$, $s$ has built some run $\rho$ over $I$ that ends in a state $q$, after which $s$ resolves a disjunction $\bigvee_{b\in \Sigma_O} \delta(q,\binom{a}{b})$ by choosing some $b\in \Sigma_O$. We then set $s'(Ia)=b$. Then, as $s$ witnesses that $\A'$ is history-deterministic, the run chosen by $s$ over an input $I\in \Sigma_I^\omega$ has the value $\best{\A}{I}$. By construction of $\A'$ and $s'$, this is the value  $\A(I\otimes \out{s'}(I)$, that is, $s'$ is indeed a solution to the synthesis problem on $\A$. Furthermore, observe that an implementation of $s$ also implements $s'$ by ignoring the outputs of $s$ that do not choose $\Sigma_O$ letters, so the memory of the solution to the synthesis problem is bounded by the memory required by a witness of history determinism.
	
	$\Longleftarrow$: Reducing the history-determinism problem to the synthesis problem.

	Dually to the previous translation, we turn the nondeterminism in an automaton $\A$ into choices of output letters in the best-value synthesis problem. We build a deterministic automaton $\A'$ that is similar to $\A$ except that it reads both an input letter and a transition; then a transition can only be chosen if it is the second element of the input (that is, the output letter). Then $\A'$ maps valid runs of $\A$ to their value and a solution to the local best value synthesis problem of $\A'$ corresponds exactly to a witness of history-determinism for $\A$.
	
	Formally, let $\A'$ be the $\Val$ automaton $(\Sigma \times \Delta, Q,\iota,\delta')$
	where $\delta'(q,(a,q\xrightarrow{a:x}q'))= (x,q')$ if $(x,q')\in \delta(q,a)$. $\A'$ maps valid runs of $\A$ written as pairs $(w,r)$ where $r$ is a run of $\A$ over $w$, onto $\Val(r)$ and in particular $\best{\A'}{I}=\A(I)$.
	
	We claim that $\A'$ is best-value realisable if and only $\A$ is history-deterministic. Indeed,  a solution $s$ to the best value synthesis problem of $\A'$ corresponds to a function building a run of $\A$ over the input $I$ transition by transition such that the value of the run is $\best{\A'}{I}$. Since $\best{\A'}{I}$ is $\A(I)$, $s$ is precisely a witness of history-determinism in $\A$. Similarly, a witness of history-determinism in $\A$ induces a solution to the best value synthesis problem for $\A'$ since it builds a run of $\A$ over $I$ with value at least $\A(I)$, exactly what is required from a solution to the best value synthesis.

\end{proof}

\end{document}